\documentclass[aps,pra,twocolumn,10pt,superscriptaddress]{revtex4-1}

\usepackage{mathrsfs}
\usepackage{bbm}
\usepackage{amsmath}
\usepackage{amssymb}
\usepackage{amsthm}
\usepackage{graphicx}
\usepackage{MnSymbol}
\usepackage{mathtools}
\usepackage{stmaryrd}
\usepackage{enumitem}
\usepackage{bbold}
\usepackage[italicdiff]{physics}
\usepackage{chemformula}

\makeatletter
\usepackage{hyperref}
\usepackage{color}
\definecolor{supcol}{RGB}{10,50,180}
\definecolor{eqcol}{RGB}{220,10,100}
\hypersetup{
	colorlinks,
	citecolor=supcol,
	linkcolor=eqcol,
	urlcolor=supcol
}

\allowdisplaybreaks

\newtheorem{theorem}{Theorem}

\newtheorem{proposition}[theorem]{Proposition}

\DeclareMathOperator{\atanh}{arctanh}
\newcommand{\mket}[1]{{|#1\rangle}}

\newcommand{\mca}{\mathcal}
\newcommand{\mbb}{\mathbb}

\newcommand{\msf}{\mathsf}

\newcommand{\msc}{\mathscr}

\begin{document}
\title{Geometric characterization for cyclic heat engines far from equilibrium}

\author{Tan Van Vu}
\email{tan.vu@riken.jp}
\affiliation{Analytical quantum complexity RIKEN Hakubi Research Team, RIKEN Center for Quantum Computing (RQC), 2-1 Hirosawa, Wako, Saitama 351-0198, Japan}

\author{Keiji Saito}
\email{keiji.saitoh@scphys.kyoto-u.ac.jp}
\affiliation{Department of Physics, Kyoto University, Kyoto 606-8502, Japan}

\date{\today}

\begin{abstract}
Considerable attention has been devoted to microscopic heat engines in both theoretical and experimental aspects. Notably, the fundamental limits pertaining to power and efficiency, as well as the tradeoff relations between these two quantities, have been intensively studied. This study aims to shed further light on the ultimate limits of heat engines by exploring the relationship between the geometric length along the path of cyclic heat engines operating at arbitrary speeds and their power and efficiency. We establish a tradeoff relation between power and efficiency using the geometric length and the timescale of the heat engine. Remarkably, because the geometric quantity comprises experimentally accessible terms in classical cases, this relation is useful for the inference of thermodynamic efficiency. Moreover, we reveal that the power of a heat engine is always upper bounded by the product of its geometric length and the statistics of energy. Our results provide a geometric characterization of the performance of cyclic heat engines, which is universally applicable to both classical and quantum heat engines operating far from equilibrium.
\end{abstract}

\maketitle

\section{Introduction}
Since the advent of macroscopic heat engines in industry, the theoretical investigation of thermodynamic heat engines has been a fundamental area of study within the realm of physics.
One well-known result is the Carnot bound, which imposes an upper bound on the efficiency of heat engines according to the second law of thermodynamics.
Beyond the macroscopic regime, our understanding of small heat engines has been significantly augmented, thanks to the recent progress in stochastic and quantum thermodynamics for microscopic systems \cite{Sekimoto.2010,Seifert.2012.RPP,Vinjanampathy.2016.CP,Deffner.2019}. 
Additionally, owing to technological advancements, it is now feasible to experimentally demonstrate the theoretical results of nanoscale heat engines on various platforms \cite{Blickle.2012.NP,Abah.2012.PRL,Martinez.2016.NP,Ronagel.2016.S,Lindenfels.2019.PRL,Klatzow.2019.PRL,Peterson.2019.PRL,Maslennikov.2019.NC,Ono.2020.PRL,Bouton.2021.NC,Zhang.2022.NC,Myers.2022.QS,Cangemi.2023.arxiv}.

The primary focus of research in the field of heat engines pertains to power, efficiency, and the inherent tradeoffs between these incompatible quantities \cite{Benenti.2011.PRL,Tomas.2013.PRE,Kosloff.2014.ARPC,Whitney.2014.PRL,Izumida.2014.PRL,Brandner.2015.PRE,Proesmans.2015.PRL,Shiraishi.2016.PRL,Brandner.2016.PRE,Proesmans.2016.PRL,Dechant.2017.EPL,Polettini.2017.EPL,Watanabe.2017.PRL,Kosloff.2017.E,Ma.2018.PRE,Pietzonka.2018.PRL,Manikandan.2019.PRL}.
The investigation of whether a heat engine can operate at maximum efficiency while still delivering finite power has been of particular interest.
Regarding this incompatibility between power and efficiency, several tradeoff relations have been uncovered for both cyclic and steady-state heat engines \cite{Shiraishi.2016.PRL,Pietzonka.2018.PRL,Menczel.2020.PRR}.
It has been demonstrated that the Carnot efficiency can only be achieved at finite power under extreme circumstances, such as diverging power fluctuations \cite{Campisi.2016.NC,Holubec.2017.PRE,Pietzonka.2018.PRL} and vanishing relaxation times \cite{Holubec.2018.PRL,Miura.2022.PRE}.

Recently, novel insights into the performance of heat engines have been obtained through the lens of information geometry \cite{Ruppeiner.1995.RMP,Abiuso.2020.PRL,Brandner.2020.PRL,Miller.2020.PRL.GWF,Vu.2021.PRL,Watanabe.2022.PRR,Eglinton.2022.PRE}.
It is noteworthy that several geometric quantities provide us useful physical interpretations that are commonly seen in various models.
For example, the thermodynamic length is interpreted as the minimum dissipation for a given path \cite{Salamon.1983.PRL,Crooks.2007.PRL,Zulkowski.2012.PRE,Sivak.2012.PRL,Machta.2015.PRL,Scandi.2019.Q}.
In the context of such a physical interpretation, the utilization of geometry in heat engines operating in the slow-driving regime results in a constraint between power and efficiency \cite{Brandner.2020.PRL,Frim.2022.PRL}, which further provides design principles for optimal parameterization of a given protocol \cite{Brandner.2020.PRL,Miller.2020.PRL.GWF,Alonso.2022.PRXQ}.
Despite the importance of geometric characterization for heat engines, previous studies on engine cycles have been limited to the near-equilibrium regime. 
Therefore, it is reasonable to ask whether geometric characterizations exist in the nonequilibrium regime and, if so, how heat engines can be physically characterized.

\begin{figure}[b]
\centering
\includegraphics[width=1.0\linewidth]{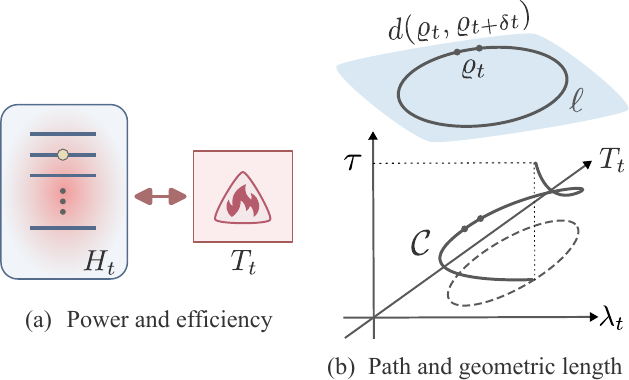}
\protect\caption{(a) Schematic of cyclic heat engines. The working medium is constantly coupled to a heat bath. The system Hamiltonian $H_t=H(\lambda_t)$ and the bath temperature $T_t$ are periodically controlled in order to extract work. (b) For a given control path $\mca{C}$, the corresponding curve of engine states $\{\varrho_t\}_{0\le t\le\tau}$ is depicted by the solid line, and its geometric length with respect to a metric $d(\cdot,\cdot)$ is denoted by $\ell$. Our aim is to elucidate the relationship between the geometric length $\ell$ and power and efficiency of heat engines.}\label{fig:Cover}
\end{figure}

In the present work, we address this open question by imposing three criteria on geometric quantities.
(i) First, geometric quantities should depend only on a predetermined path, denoted as $\mca{C}$, of a heat-engine cycle, (ii) second, they should possess physical interpretations, and (iii) third, they should be experimentally computable at least in classical cases.
Concerning the path dependency, a schematic illustration extending from the quasi-static case to a non-equilibrium scenario is shown in Fig.~\ref{fig:Cover}(b). 
As for the second criterion, we consider the geometric length defined on the space of state probabilities (density matrices in quantum cases) of the heat engine, effectively quantifying the degree of state change [cf.~Eq.~\eqref{metric.def}]. 
In the stationary state, the geometric length is solely dependent on the path.
We impose the third criterion, as estimable geometric quantities can facilitate thermodynamic inference of efficiency, especially in classical cases, without requiring detailed knowledge of the system. 
It is important to note that in experiments, state probabilities and jump statistics are measurable, whereas accurately measuring heat flow via calorimetry and determining precise dissipative dynamics are generally challenging tasks.

Considering heat engines that obey Markovian dissipative dynamics, we first derive a geometric tradeoff between power and efficiency [cf.~Eq.~\eqref{eq:power.eff.tradeoff}].
This relation means that the product of power and the difference between efficiency and Carnot efficiency cannot be arbitrarily small if the geometric length and timescale of heat engines are predetermined. 
This also provides a useful tool for inferring efficiency, as the geometric bound is written solely in terms of experimentally accessible quantities in classical cases.
Subsequently, we prove that power is always upper bounded by the geometric length, energy fluctuation, and energy gap [cf.~Eq.~\eqref{eq:power.bound}].
These findings are fundamental and can be applied to both classical and quantum heat engines.

\section{Setup}
\subsection{Cyclic heat engines}
We consider a heat engine using a finite-dimensional discrete system as the working medium. 
The heat engine is coupled to a heat bath and is driven by a periodic control protocol.
To study both classical and quantum dissipative heat engines in a unified manner, we employ the Lindblad master equation \cite{Lindblad.1976.CMP,Gorini.1976.JMP}. 
This equation describes the time evolution of the density matrix of the system as follows:
\begin{equation}\label{eq:mas.eq}
\dot\varrho_t=-i\hbar^{-1}[H_t,\varrho_t]+\sum_k\mca{D}[L_k(t)]\varrho_t,
\end{equation}
where $H_t$ is the controllable Hamiltonian, $\mca{D}[L]\varrho\coloneqq L\varrho L^\dagger-\{L^\dagger L,\varrho\}/2$ is the dissipator, and $\{L_k(t)\}$ are the Lindblad jump operators.
To guarantee the thermodynamic consistency, we assume the local detailed balance condition; that is, the jump operators come in pairs $(k,k')$ such that $L_k(t)=e^{\beta_t\omega_k(t)/2}L_{k'}(t)^\dagger$, where $\beta_t=1/(k_BT_t)$ is the inverse temperature and $\omega_k(t)=-\omega_{k'}(t)$ denotes the energy change associated with the $k$th jump.
Hereafter, we set both $\hbar$ and $k_B$ to unity for simplicity.
The temperature of the heat bath can be generically modulated as $T_t={T_cT_h}/{[T_h+\alpha_t(T_c-T_h)]}$ \cite{Brandner.2015.PRX,Brandner.2016.PRE}, where $T_c$ and $T_h$ are the minimum and maximum temperatures reached by the heat bath, respectively, and $\alpha_t$ is a periodic function satisfying $\min_t\alpha_t=0$ and $\max_t\alpha_t=1$.
The work power extractable from the heat engine can be calculated as
\begin{equation}\label{eq:power.def}
P_w\coloneqq-\tau^{-1}\int_0^\tau\tr{\dot{H}_t\varrho_t}\dd{t}.
\end{equation}
The irreversible entropy production of the heat engine can be defined as the sum of changes in the system and environmental entropy, given by
\begin{equation}\label{eq:ent.prod}
  \Sigma_\tau\coloneqq\Delta S-  \int_0^\tau\beta_tQ_t\dd{t},
\end{equation}
where $\Delta S\coloneqq\tr{\varrho_0\ln\varrho_0}-\tr{\varrho_\tau\ln\varrho_\tau}$ is the production of the von Neumann entropy and $Q_t\coloneqq\tr{H_t\dot\varrho_t}$ is the rate of heat absorbed by the engine.
By exploiting the periodicity, we obtain $\Delta S=0$ and arrive at the following expression of entropy production:
\begin{equation}\label{eq:ent.prod.exp}
\Sigma_\tau=\frac{\tau}{T_c}\qty(-P_w+\eta_CJ_q)\ge 0,
\end{equation}
where $\eta_C\coloneqq 1-T_c/T_h$ is the Carnot efficiency and $J_q\coloneqq\tau^{-1}\int_0^\tau\alpha_tQ_t\dd{t}$ is the time-average heat supplied to the engine.
Expression \eqref{eq:ent.prod.exp} naturally leads to the following definition of the efficiency of cyclic heat engines \cite{Brandner.2016.PRE,Miller.2021.PRL.TUR}:
\begin{equation}
\eta\coloneqq\frac{P_w}{J_q}\le\eta_C.
\end{equation}
In the case where $\alpha_t$ is the step function, the conventional setup, wherein the system is coupled to two distinct heat baths with respectively constant temperatures $T_c$ and $T_h$, is recovered.
Another quantity of interest is the timescale of heat engines, which can be characterized by the time-average number of jumps:
\begin{equation}
\overline{a}_\tau\coloneqq\tau^{-1}\int_0^\tau\sum_k\tr{L_k(t)\varrho_tL_k(t)^\dagger}\dd{t}.
\end{equation}
This quantity is also known as dynamical activity in the literature \cite{Maes.2020.PR} and plays an important role in constraining speed and fluctuation of nonequilibrium systems \cite{Shiraishi.2018.PRL,Hasegawa.2020.PRL,Vu.2022.PRL.TUR}.
In quantum cases, dynamical activity is the average number of quantum jumps \cite{Gleyzes.2007.N,Vijay.2011.PRL,Vool.2014.PRL,Vool.2016.PRL}. 
In classical cases, dynamical activity can be easily calculated by taking the ensemble average of the number of jumps between different states over the stochastic trajectories in experiments.

\subsection{Relevant metrics}
Next, we discuss several distance measures that will be used in this study.
Given a curve of density matrices $\{\varrho_t\}_{0\le t\le\tau}$ on a control path $\mca{C}$, its geometric length with respect to an arbitrary metric $d(\cdot,\cdot)$ can be defined as
\begin{equation}
\ell(\mca{C})\coloneqq\lim_{\delta t\to 0}\sum_{k=0}^{K-1}d(\varrho_{k\delta t},\varrho_{(k+1)\delta t}), \label{metric.def}
\end{equation}
where $\delta t\coloneqq\tau/K$ is the discretized time interval.
Relevant metrics that will be used for $d(\varrho,\sigma)$ are a Wasserstein distance $W(\varrho,\sigma)\coloneqq\sum_n|a_n-b_n|/2$ \cite{Vu.2023.PRX}, where $\{a_n\}$ and $\{b_n\}$ are increasing eigenvalues of $\varrho$ and $\sigma$, respectively, and the trace distance $T(\varrho,\sigma)\coloneqq\tr{|\varrho-\sigma|}/2$.
The classical versions of these measures are used for classical cases.
When $d$ is a Riemannian metric defined via an inner product $g_\varrho$, the curve length can also be calculated as $\ell=\int_0^\tau\sqrt{g_{\varrho_t}(\dot{\varrho}_t,\dot{\varrho}_t)}\dd{t}$.
Other commonly used metrics are the contractive Riemannian metrics defined on the space of density matrices, which can be completely characterized by Morozova-Chentsov (MC) functions \cite{Morozova.1991.JSM,Petz.1996.LAA} and are defined in the following form:
\begin{equation}
g_{\varrho}(X,Y)\coloneqq\tr{X[\msc{R}_\varrho f(\msc{L}_\varrho\msc{R}_\varrho^{-1})]^{-1}Y}.
\end{equation}
Here, $X$ and $Y$ are self-adjoint traceless matrices, $\msc{L}_\varrho$ and $\msc{R}_\varrho$ are two super-operators, defined as $\msc{L}_\varrho(X)\coloneqq\varrho X$ and $\msc{R}_\varrho(X)\coloneqq X\varrho$, and $f(t)$ is a MC function satisfying three conditions: (1) $f(X)\preceq f(Y)$ for any semipositive-definite operators $X \preceq Y$, (2) $f(t)=tf(t^{-1})$, and (3) $f(1)=1$. 
It has been shown that a MC function always fulfills $f_{\rm min}(t)\le f(t)\le f_{\rm max}(t)$ \cite{Kubo.1980.MA}, where $f_{\rm min}(t)=2t/(1+t)$ and $f_{\rm max}(t)=(1+t)/2$ are the minimum and maximum MC functions, respectively.
Hereafter, we exclusively use the right logarithmic derivative Fisher information, which is induced by the minimum MC function $f_{\rm min}(t)$. This quantum Fisher information metric can be explicitly expressed as $g_\varrho(\dot{\varrho},\dot{\varrho})=\tr{\dot{\varrho}^2\varrho^{-1}}$.
We emphasize that these metrics provide us the physical interpretations, i.e., the degree of state change along the path $\mca{C}$.

\section{Results}
We are now ready to explain two main results, whose proofs are postponed to the end of the section.
Our first main result is a fundamental tradeoff between power and efficiency of cyclic heat engines. We prove that
\begin{align}
P_w\qty(\frac{\eta_C}{\eta}-1) &\ge g(\mca{C}),\label{eq:power.eff.tradeoff}\\
g(\mca{C})&\coloneqq 2T_c\frac{\ell_{W}(\mca{C})}{\tau}\atanh\qty[\frac{\ell_{W}(\mca{C})}{\tau\overline{a}_\tau}],
\end{align}
where $\ell_{W}$ denotes the geometric length with respect to the quantum Wasserstein distance.
For classical cases, $\ell_{W}$ is replaced by the geometric length with the classical version of the Wasserstein distance (see Appendix \ref{app:cla.tradeoff} for more details).
Note that the classical Wasserstein metric $W$ can be easily calculated using the linear programming method \cite{Vu.2023.PRL.TSL}. 
Therefore, the geometric length $\ell_{W}$ is experimentally computable because $d\ell_{W}=W(\varrho_t,\varrho_{t+dt})$ and state probabilities can be accurately obtained in experiments.

Some remarks on inequality \eqref{eq:power.eff.tradeoff} are given in order. 
First, $g(\mca{C})$ comprises experimentally accessible quantities, especially in classical cases, regardless of the model's specific details. 
As such, it satisfies all the criteria (i)-(iii) in the introduction. 
In classical cases, measuring both the state probability and dynamical activity is considerably easier compared to the calorimetric measurement of heat flow. 
Therefore, inequality \eqref{eq:power.eff.tradeoff} can be utilized for thermodynamic inference.
Second, it indicates a novel tradeoff between power and efficiency; that is, both $P_w$ and $\eta_C/\eta-1$ cannot be simultaneously small, given that the geometric length, dynamical activity, and time duration are fixed.
Third, it provides a comprehensive characterization about the attainability of Carnot efficiency at finite power.
By rearranging Eq.~\eqref{eq:power.eff.tradeoff}, an upper bound on efficiency can be obtained as
\begin{equation}\label{eq:eff.upb}
	\eta\le\frac{\eta_C}{\mca{B}+1},
\end{equation}
where $\mca{B}\coloneqq 2T_c\ell_{W}\atanh\qty({\ell_{W}}/{\tau\overline{a}_\tau})/(\tau P_w)$.
Evidently, $\eta\to\eta_C$ only if $\mca{B}\to 0$.
This essentially provides a condition for attaining Carnot efficiency at finite power and finite time: $\tau\overline{a}_\tau\gg\ell_W$ \cite{fnt1}.
In other words, the dynamical activity must diverge relative to the curve length.
Since dynamical activity is related to the inverse of the relaxation time of heat engines, the divergence of dynamical activity is connected to the vanishing of the relaxation time, which is consistent with results reported previously for Langevin dynamics \cite{Miura.2022.PRE}.
In the quantum regime, the divergence of dynamical activity can be achieved by exploiting coherence between degenerate energy eigenstates; as a result, power can be scaled as $O(N)$ while $\eta_C/\eta-1=O(1/N)$, where $N$ is the number of degeneracy in the system Hamiltonian \cite{Tajima.2021.PRL}.
Fourth, in the adiabatic regime, the tradeoff relation leads to the following geometric bound on efficiency (refer to Appendix \ref{app:adia.tradeoff} for more details):
\begin{equation}\label{eq:tradeoff.adia}
	\eta\le\eta_C\qty(1-\frac{2T_c}{\overline{a}_\tau^\infty}\frac{(\ell_W^\infty)^2}{\mca{W}\tau}).
\end{equation}
Here, $\mca{W}$, $\ell_W^\infty$, and $\overline{a}_\tau^\infty$ are the quasistatic work, geometric length, and dynamical activity, respectively \cite{fnt2}.
It is noteworthy that upper bound \eqref{eq:tradeoff.adia} is in terms of only quantities at equilibriums, whereas the extant bound \cite{Brandner.2020.PRL} involves a response quantity from equilibriums.

Last, we notice that Eq.~\eqref{eq:power.eff.tradeoff} provides a lower bound on the power, and combining that with an extant relation yields
\begin{align}
  g(\mca{C})\qty(\frac{\eta_C}{\eta}-1)^{-1}\le P_w \le \Theta T_c^{-1}\eta(\eta_C-\eta), 
\end{align}
where the upper bound was given in the previous study \cite{Shiraishi.2016.PRL}.
Here, $\Theta$ is a system-dependent quantity for cyclic engines.
Most studies thus far have provided the upper bounds for the power \cite{Shiraishi.2016.PRL,Pietzonka.2018.PRL}.
Therefore, our relation can be considered a complementary tradeoff relation for heat engines.

Next, we present our second main result, which addresses the ultimate limit on the extractable power of cyclic heat engines, given a fixed geometric length of the curve of density matrices. 
Using the quantum Fisher information metric and the trace distance to quantify the geometric length, we prove that the power is always upper bounded by the geometric length and the energy statistics as
\begin{equation}\label{eq:power.bound}
P_w\le\min\qty{\frac{\ell_{F}(\mca{C})}{\tau}\max_t\Delta H_t,\frac{\ell_{T}(\mca{C})}{\tau}\max_t\|H_t\|_*}.
\end{equation}
Here, $\ell_{F}$ and $\ell_T$ denote the geometric length associated with the quantum Fisher information metric and the trace distance, respectively, $(\Delta H_t)^2\coloneqq\tr\{H_t^2\varrho_t\}-\tr\{H_t\varrho_t\}^2$ is the energy fluctuation, and $\|H_t\|_*\coloneqq\max_{m,n}|\varepsilon_m(t)-\varepsilon_n(t)|$ is the energy gap, where $\{\varepsilon_n(t)\}$ are eigenvalues of $H_t$.
Since inferring the energy gap and energy fluctuation through experiments can be challenging, it is appropriate to argue that Eq.~\eqref{eq:power.bound} fulfills only criteria (i) and (ii) outlined in the introduction. Nevertheless, it provides a physically essential implication: to achieve a high power output in cyclic heat engines for a given geometric length, it is necessary to increase both the energy fluctuation and the energy gap.
This implication is somewhat reminiscent of speed limits in closed quantum systems, in which fast state transformations analogously require large energy fluctuation and energy gap \cite{Mandelstam.1945.JP,Margolus.1998.PD}. 
If these energy statistics are finite, the geometric length limits the power.

\subsection{Numerical demonstration}
We exemplify our results using a superconducting qubit heat engine, which is a simple model of the recent experimental realization \cite{Ono.2020.PRL}.
The heat engine is described by the tunable Hamiltonian $H_t=\varepsilon_t[\sin(\theta_t)\sigma_x+\cos(\theta_t)\sigma_z]/2$ and jump operators $L_1(t)=\sqrt{\gamma\varepsilon_t(n_t+1)}\dyad{0_t}{1_t}$ and $L_2(t)=\sqrt{\gamma\varepsilon_t n_t}\dyad{1_t}{0_t}$.
Here, $\gamma$ is the coupling strength, $\varepsilon_t$ and $\theta_t$ characterize the instantaneous energy gap and the relative strength of coherent tunneling to energy bias, respectively, $n_t\coloneqq 1/(e^{\beta_t\varepsilon_t}-1)$, and $\{\ket{0_t},\ket{1_t}\}$ denote the instantaneous eigenstates of the Hamiltonian.
The control protocol is given by $\{\alpha_t,\varepsilon_t,\theta_t\}$, where $\alpha_t=\sin(\pi t/\tau)^2$, $\varepsilon_t=\Omega [1+0.5\sin(2\pi t/\tau)+0.1\sin(6\pi t/\tau)]$, and $\theta_t=\pi[1+0.1\sin(\pi t/\tau)^2]/2$.

\begin{figure}[t]
\centering
\includegraphics[width=1.0\linewidth]{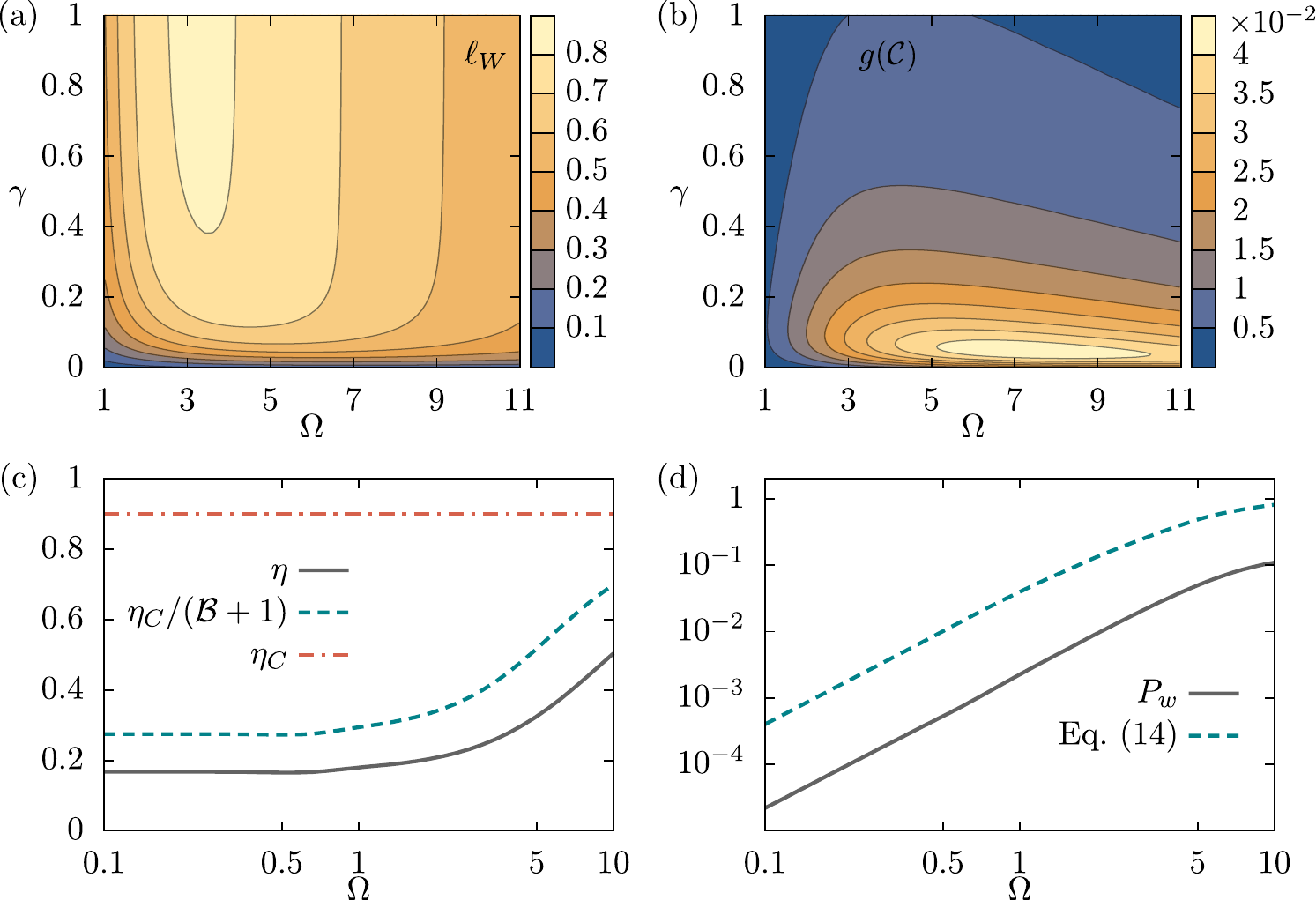}
\protect\caption{Numerical illustration for the qubit heat engine. The top panels show contour plots of (a) the geometric length $\ell_W$ and (b) the lower bound in Eq.~\eqref{eq:power.eff.tradeoff} on the $\Omega$-$\gamma$ plane. The bottom panels describe upper bounds of (c) efficiency [Eq.~\eqref{eq:eff.upb}] and (d) power [Eq.~\eqref{eq:power.bound}]. Parameter $\Omega$ is varied while other parameters are fixed as $T_c=1$, $T_h=10$, $\gamma=0.1$, and $\tau=10$.}\label{fig:Numerical}
\end{figure}

To help clarify the behavior of the geometric length $\ell_W$ and the lower bound in Eq.~\eqref{eq:power.eff.tradeoff}, contour plots of these quantities on the $\Omega$-$\gamma$ plane are shown in Figs.~\ref{fig:Numerical}(a) and \ref{fig:Numerical}(b), respectively.
As seen in Fig.~\ref{fig:Numerical}(a), the geometric length $\ell_W$ monotonically increases with respect to the coupling strength $\gamma$, which agrees with our intuition.
By contrast, $\ell_W$ is nonmonotonic with respect to $\Omega$.
That is, there always exists an optimal $\Omega$ that yields the maximum value of $\ell_W$, provided that $\gamma$ is fixed.
Figure \ref{fig:Numerical}(b) shows the landscape of the lower bound in Eq.~\eqref{eq:power.eff.tradeoff}, where the product $P_w(\eta_C/\eta-1)$ cannot be below.
It turns out that there is only one extreme point of $(\Omega,\gamma)$ that achieves the maximum of the lower bound. We next numerically verify the derived bounds [Eqs.~\eqref{eq:eff.upb} and \eqref{eq:power.bound}].
Parameter $\gamma$ is fixed whereas $\Omega$ is varied, and numerical results are plotted in Figs.~\ref{fig:Numerical}(c) and \ref{fig:Numerical}(d).
As shown in Fig.~\ref{fig:Numerical}(c), the bound \eqref{eq:eff.upb} yields a tighter restriction for efficiency $\eta$ than the conventional Carnot bound.
Figure \ref{fig:Numerical}(d) illustrates that power $P_w$ grows monotonically as parameter $\Omega$ (which characterizes the energy gap) is increased from $0.1$ to $10$, which is consistent with the qualitative statement deduced from Eq.~\eqref{eq:power.bound}.
The validity of the bounds can thus be confirmed from the figures.

\subsection{Proof of Eqs.~\eqref{eq:power.eff.tradeoff} and \eqref{eq:power.bound}}
First, we prove Eq.~\eqref{eq:power.eff.tradeoff}.
Let $\varrho_t=\sum_np_n(t)\dyad{n_t}$ be the spectral decomposition of the density matrix $\varrho_t$.
We then define transition rates between the eigenbasis as $r^k_{mn}(t)\coloneqq|\mel{m_t}{L_k(t)}{n_t}|^2\ge 0$.
Notice that $r^k_{mn}(t)=e^{\beta_t\omega_k(t)}r^{k'}_{nm}(t)$.
Taking the time derivative of $p_n(t)=\mel{n_t}{\varrho_t}{n_t}$, we can derive the following master equation for the distribution $\{p_n(t)\}$ from Eq.~\eqref{eq:mas.eq}:
\begin{equation}\label{eq:eigen.master.eq}
\dot{p}_n(t)=\sum_k\sum_{m(\neq n)}\qty[r^{k}_{nm}(t)p_m(t)-r^{k'}_{mn}(t)p_n(t)].
\end{equation}
For convenience, we define $a^{k}_{mn}(t)\coloneqq r^k_{mn}(t)p_n(t)$ and $j^{k}_{mn}(t)\coloneqq r^k_{mn}(t)p_n(t)-r^{k'}_{nm}(t)p_m(t)$.
Using these probability currents, we can express the master equation \eqref{eq:eigen.master.eq} as $\dot p_n(t)=\sum_k\sum_{m(\neq n)}j^k_{nm}(t)$.
The entropy production rate ${\sigma}_t\coloneqq \dot{\Sigma}_t$ and dynamical activity rate ${a}_t\coloneqq\partial_t[t\overline{a}_t]$ can be analytically expressed as
\begin{align}
\sigma_t&=\frac{1}{2}\sum_k\sum_{m,n}j^{k}_{mn}(t)\ln\frac{a^{k}_{mn}(t)}{a^{k'}_{nm}(t)}\eqqcolon\frac{1}{2}\sum_k\sum_{m,n}\sigma^{k}_{mn}(t),\label{eq:qua.ent.prod.rate}\\
a_t&=\frac{1}{2}\sum_k\sum_{m,n}[a^{k}_{mn}(t)+a^{k'}_{nm}(t)].
\end{align}
By noting that $x\Phi(x/y)^{-1}$ is a concave function, where $\Phi(x)$ is the inverse function of $x\tanh(x)$, and applying Jensen's inequality and the triangle inequality, we can prove that \cite{Vu.2023.PRX}
\begin{align}
\frac{\sigma_t}{2}\Phi\qty(\frac{\sigma_t}{2a_t})^{-1}&\ge\sum_{k,m,n}\frac{\sigma^{k}_{mn}(t)}{4}\Phi\qty(\frac{\sigma^{k}_{mn}(t)}{2[a^{k}_{mn}(t)+a^{k'}_{nm}(t)]})^{-1}\notag\\
&=\frac{1}{2}\sum_{k,m,n}|j^{k}_{mn}(t)|\notag\\
&\ge\frac{1}{2}\sum_n|\dot{p}_n(t)|.\label{eq:tmp2}
\end{align}
Taking the time integration of Eq.~\eqref{eq:tmp2} from $t$ to $t+\delta t$ and noting that $(1/2)\int_t^{t+\delta t}\sum_n|\dot{p}_n(s)|\dd{s}\ge W(\varrho_t,\varrho_{t+\delta t})$, we immediately obtain
\begin{equation}\label{eq:tmp1}
\int_t^{t+\delta t}\frac{\sigma_s}{2}\Phi\qty(\frac{\sigma_s}{2a_s})^{-1}\dd{s}\ge W(\varrho_t,\varrho_{t+\delta t}).
\end{equation}
Here, $\delta t>0$ is an arbitrary constant.
By summing both sides of Eq.~\eqref{eq:tmp1} for $t=k\delta t~(k=0,\dots,K-1)$, where $K=\tau/\delta t$, and exploiting again the concavity of $x\Phi(x/y)^{-1}$, we arrive at the following inequality:
\begin{align}
\sum_{k=0}^{K-1}W(\varrho_{k\delta t},\varrho_{(k+1)\delta t})&\le \frac{\Sigma_\tau}{2}\Phi\qty(\frac{\Sigma_\tau}{2\tau\overline{a}_\tau})^{-1}.
\end{align}
Taking the $\delta t\to 0$ limit, using the monotonicity of functions $x\tanh(x)$ and $\atanh(x)$, and rearranging the inequality $\ell_{W}\le (\Sigma_\tau/2)\Phi\qty(\Sigma_\tau/2\tau\overline{a}_\tau)^{-1}$ immediately yields Eq.~\eqref{eq:power.eff.tradeoff}.

Next, we prove Eq.~\eqref{eq:power.bound}.
By applying the Cauchy-Schwarz inequality $\tr{AB}\le\tr\{A^\dagger A\}^{1/2}\tr\{B^\dagger B\}^{1/2}$ for arbitrary matrices $A$ and $B$, the first argument of Eq.~\eqref{eq:power.bound} can be proved as follows:
\begin{align}
\tau P_w&=\int_0^\tau\tr{(H_t-\ev{H_t})\dot{\varrho}_t}\dd{t}\notag\\
&\le\int_0^\tau\tr{(H_t-\ev{H_t})^2\varrho_t}^{1/2}\tr{\dot{\varrho}_t^2\varrho_t^{-1}}^{1/2}\dd{t}\notag\\
&=\int_0^\tau \Delta H_t \sqrt{g_{\varrho_t}(\dot{\varrho}_t,\dot{\varrho}_t)}\dd{t}\notag\\
&\le \ell_{F}(\mca{C})\max_{t}\Delta H_t.
\end{align}
Next, we prove the second argument of Eq.~\eqref{eq:power.bound}.
Note that the trace distance can be expressed in the following variational form:
\begin{equation}\label{eq:trace.dist.var.form}
T(\varrho,\sigma)=\max_{\mbb{0}\le X\le\mbb{1}}\tr{X(\varrho-\sigma)},
\end{equation}
where the maximum is over all self-adjoint matrices $X$ satisfying the condition. 
Noticing $\|H_t\|_*=\max_n\varepsilon_n(t)-\min_n\varepsilon_n(t)$ and defining $\epsilon_t\coloneqq\min_n\varepsilon_n(t)$ and $X_t\coloneqq\|H_t\|_*^{-1}(H_t-\epsilon_t\mbb{1})$, we can easily show that $\mbb{0}\le X_t\le \mbb{1}$.
By applying the variational formula \eqref{eq:trace.dist.var.form}, we can prove as follows:
\begin{align}
\tau P_w&=\int_0^\tau\tr{(H_t-\epsilon_t\mbb{1})\dot{\varrho}_t}\dd{t}\notag\\
&=\int_0^\tau\|H_t\|_*\tr{X_t\dot{\varrho}_t}\dd{t}\notag\\
&\le\int_0^\tau\|H_t\|_*\frac{T(\varrho_t,\varrho_{t+dt})}{dt}\dd{t}\notag\\
&\le\ell_{T}(\mca{C})\max_t\|H_t\|_*.
\end{align}

\section{Conclusion}
We derived geometric bounds for the power and efficiency of cyclic heat engines driven at arbitrary speeds.
Specifically, we revealed a geometric tradeoff relation between power $P_w$ and efficiency term $\eta_C/\eta-1$, indicating that these quantities cannot be simultaneously small, provided that the geometric length and timescale of heat engines are predetermined.
This tradeoff relation not only provides insight into the performance of heat engines but also offers a useful method for estimating their efficiency.
In addition, we presented qualitative and quantitative constraints on power in terms of geometric length and energy statistics, showing that attaining high power output requires large energy gaps and energy fluctuations.

\begin{acknowledgments}
The authors thank Kay Brandner and Andreas Dechant for fruitful discussions. 
T.V.V.~was supported by JSPS KAKENHI Grant No.~JP23K13032.
K.S.~was supported by JSPS KAKENHI Grant No.~JP23H01099, No.~JP19H05603, and No.~JP19H05791.
\end{acknowledgments}

\appendix

\section{Classical trade-off relation between power and efficiency}\label{app:cla.tradeoff}
Here we provide the classical analog of the quantum trade-off relation \eqref{eq:power.eff.tradeoff} presented in the main text.
Note that in the classical case, the dynamics of the heat engine is described by the master equation:
\begin{equation}
\ket{\dot{p}_t}=\msf{R}_t\ket{p_t},
\end{equation}
where $\ket{p_t}\coloneqq[p_1(t),\dots,p_N(t)]^\top$ is the probability distribution of the system and $\msf{R}_t=[r_{mn}(t)]\in\mbb{R}^{N\times N}$ is the transition rate matrix with $r_{nn}(t)=-\sum_{m(\neq n)}r_{mn}(t)$.
The jump process is assumed to be microscopically reversible, that is $r_{mn}(t)>0$ whenever $r_{nm}(t)>0$.
For the sake of thermodynamic consistency, we also assume the detailed balance condition $r_{mn}(t)e^{-\beta_t\varepsilon_n(t)}=r_{nm}(t)e^{-\beta_t\varepsilon_m(t)}$, where $\varepsilon_n(t)$ denotes the instantaneous energy level of state $n$ and $\beta_t\coloneqq 1/(k_BT_t)$ is the inverse temperature.

In the classical case, we consider the discrete $L^1$-Wasserstein distance, which is defined based on the connectivity of the Markov jump process.
Consider a graph consisting of $N$ vertices. An edge between vertices $m$ and $n$ is present in the graph whenever there exists a transition between states $m$ and $n$ [i.e., $r_{mn}(t)>0$].
Let $\msf{D}=[d_{mn}]$ be the matrix of the shortest-path distances between vertices of the graph. Here, the length of a path is quantified by the number of edges contained in the path. Then, the $L^1$-Wasserstein distance can be defined in the following variational form \cite{Dechant.2022.JPA,Vu.2023.PRX}:
\begin{equation}
W(p,q)\coloneqq\min_{\pi}\tr{\msf{D}^\top\pi},
\end{equation}
where the minimum is over all joint probability distributions $\{\pi\}$ of which marginal distributions coincide with $p$ and $q$.
In the case of complete connectivity (i.e., when $d_{mn}=1$ for any $m\neq n$), the Wasserstein distance reduces to the total variation distance $T(p,q)\coloneqq\sum_n|p_n-q_n|/2$.
In general, we have $W(p,q)\ge T(p,q)$.

For convenience, we define the jump frequency, the probability current, and the entropy production rate associated with the transition between $m$ and $n$ as
\begin{align}
	a_{mn}(t)&\coloneqq r_{mn}(t)p_n(t),\\
	j_{mn}(t)&\coloneqq a_{mn}(t)-a_{nm}(t),\\
	\sigma_{mn}(t)&\coloneqq j_{mn}(t)\ln[{a_{mn}(t)}/{a_{nm}(t)}].
\end{align}
By these definitions, the dynamical activity and irreversible entropy production can be expressed as
\begin{align}
	\overline{a}_\tau&=\tau^{-1}\int_0^\tau\sum_{m\neq n}a_{mn}(t)\dd{t},\\
	\Sigma_\tau&=\int_0^\tau\sum_{m>n}\sigma_{mn}(t)\dd{t}.
\end{align}
It was shown that the classical Wasserstein distance can be upper bounded by the probability currents as \cite{Vu.2023.PRX}
\begin{equation}\label{eq:Wdist.ine1}
W(p_t,p_{t+\delta t})\le\int_t^{t+\delta t}\sum_{m>n}|j_{mn}(s)|\dd{s}.
\end{equation}
Summing both sides of Eq.~\eqref{eq:Wdist.ine1} for all $t=k\delta t\,(0\le k\le K-1)$ and taking the $\delta t\to 0$ limit yields
\begin{equation}\label{eq:LW.up}
\ell_{W}\le\int_0^\tau\sum_{m>n}|j_{mn}(t)|\dd{t}.
\end{equation}
By simple algebraic calculations, we can show that \cite{Vu.2023.PRX}
\begin{equation}
|j_{mn}(t)|=\frac{\sigma_{mn}(t)}{2}\Phi\qty(\frac{\sigma_{mn}(t)}{2[a_{mn}(t)+a_{nm}(t)]})^{-1},
\end{equation}
where $\Phi(x)$ is the inverse function of $x\tanh(x)$.
Since $x\Phi(x/y)^{-1}$ is a concave function over $(0,+\infty)\times(0,+\infty)$, applying Jensen's inequality yields
\begin{align}
\int_0^\tau\sum_{m>n}|j_{mn}(t)|\dd{t}&\le\frac{\Sigma_\tau}{2}\Phi\qty(\frac{\Sigma_\tau}{2\tau\overline{a}_\tau})^{-1}.\label{eq:Jt.up}
\end{align}
Combining Eqs.~\eqref{eq:LW.up} and \eqref{eq:Jt.up}, using the monotonicity of functions $x\tanh(x)$ and $\atanh(x)$, and performing expression transformation results in the following inequality:
\begin{equation}\label{eq:ent.prod.lb}
\Sigma_\tau\ge 2\ell_{W}\atanh\qty(\frac{\ell_{W}}{\tau\overline{a}_\tau}).
\end{equation}
Using Eq.~\eqref{eq:ent.prod.lb} and the expression of irreversible entropy production, we readily obtain the following trade-off relation:
\begin{equation}\label{eq:app.cla.tradeoff}
P_w\qty(\frac{\eta_C}{\eta}-1)\ge 2T_c\frac{\ell_W}{\tau}\atanh\qty(\frac{\ell_W}{\tau\overline{a}_\tau}).
\end{equation}

\section{Trade-off relation of heat engines in the adiabatic regime}\label{app:adia.tradeoff}
We consider the case that the heat engine is operating in the adiabatic regime (i.e., $\tau\gg 1$).
Let $\{\hat{\vb*{\Lambda}}_t=[\hat T_t,\hat \varepsilon_t,\hat \theta_t]\}_{0\le t\le 1}$ be a fixed reference protocol. The finite-time protocol is described by $H_t=H({\vb*{\Lambda}}_{t})$, $T_t=T({\vb*{\Lambda}}_{t})$, and $L_k(t)=L_k({\vb*{\Lambda}}_{t})$, where $\vb*{\Lambda}_t=\hat{\vb*{\Lambda}}_{t/\tau}$.
For convenience, we denote by $v^\infty$ the scalar quantity $v$ associated with the heat engine at the $\tau\to\infty$ limit.
It is convenient to define time-rescaled quantities
\begin{equation}
\hat{\varrho}_t=\varrho_{\tau t},~\hat{\mca{L}}_t=\mca{L}_{\tau t}.
\end{equation}
The Lindblad master equation can be rewritten as $\dot{\hat{\varrho}}_t=\tau\hat{\mca{L}}_t(\hat\varrho_t)$ ($0\le t\le 1$).
It should be noted that the super-operator $\hat{\mca{L}}_t$ is independent of $\tau$, while the time-rescaled solution $\hat{\varrho}_t$ is not.
As $\tau\to\infty$, the density matrix of the system remains close to the instantaneous equilibrium state.
Therefore, $\hat{\varrho}_t$ can be expanded in terms of $1/\tau$ as \cite{Cavina.2017.PRL}
\begin{equation}
\hat\varrho_t=\hat\varrho_{0,t}+\frac{1}{\tau}\hat\varrho_{1,t}+\frac{1}{\tau^2}\hat\varrho_{2,t}+\cdots=\pi_{\hat{\vb*{\Lambda}}_t}+O(\tau^{-1}).
\end{equation}
Here, $\pi_{\vb*{\Lambda}}\coloneqq e^{-H({\vb*{\Lambda}})/T({\vb*{\Lambda}})}/\tr e^{-H({\vb*{\Lambda}})/T({\vb*{\Lambda}})}$ is the Gibbs thermal state and $\{\hat{\varrho}_{n,t}\}_{n\ge 1}$ are traceless operators.
The operators $\{\hat\varrho_{n,t}\}$ satisfy the following equation:
\begin{equation}
	\dot{\hat\varrho}_{n,t}=\hat{\mca{L}}_t\qty[\hat\varrho_{n+1,t}]~\forall n\ge 0.
\end{equation}
It is thus evident that $\hat\varrho_t=\hat\varrho_{0,t}+O(\tau^{-1})$.
We assume that both $\{\hat p_n(t)\}$ and $\{\ket{\hat n_t}\}$ (eigenvalues and eigenvectors of $\hat\varrho_t$) can be expressed in perturbative forms as $\hat p_n(t)=\hat p_{0,n}(t)+\sum_{k=1}^\infty\tau^{-k}\hat p_{k,n}(t)$ and $\ket{\hat n_t}=\mket{\hat n_{0,t}}+\sum_{k=1}^\infty\tau^{-k}\mket{\hat n_{k,t}}$, where $\{\hat p_{0,n}(t)\}$ and $\{\mket{\hat n_{0,t}}\}$ are eigenvalues and eigenvectors of $\pi_{\hat{\vb*{\Lambda}}_t}$, respectively.
The first-order terms of these perturbative expressions can be explicitly calculated using $\hat\varrho_{n,t}$.
For example, by collecting the terms of order $O(\tau^{-1})$ in relation $\hat\varrho_t\ket{\hat n_t}=\hat p_n(t)\ket{\hat n_t}$, we obtain
\begin{align}
	\hat\varrho_{0,t}\mket{\hat n_{1,t}}+\hat\varrho_{1,t}\mket{\hat n_{0,t}}=\hat p_{0,n}(t)\mket{\hat n_{1,t}}+\hat p_{1,n}(t)\mket{\hat n_{0,t}}.\label{eq:per.exp.tmp1}
\end{align}
Multiplying $\bra{\hat n_{0,t}}$ to both sides of Eq.~\eqref{eq:per.exp.tmp1} and noticing that $\hat\varrho_{0,t}\ket{\hat n_{0,t}}=\hat p_{0,n}(t)\ket{\hat n_{0,t}}$ yields an explicit form for the first-order term of eigenvalues:
\begin{equation}
	\hat p_{1,n}(t)=\mel{\hat n_{0,t}}{\hat\varrho_{1,t}}{\hat n_{0,t}}.
\end{equation}
We next prove the following arguments: $\ell_W=\ell_W^\infty+O(\tau^{-1})$ and $\overline{a}_\tau=\overline{a}_\tau^\infty+O(\tau^{-1})$.
We now verify the first argument.
As proved in Proposition \ref{prop:geo.len.exp}, the geometric length $\ell_W$ can be expressed as
\begin{equation}
	\ell_W=\frac{1}{2}\int_0^\tau\sum_n|\dot p_n(t)|\dd{t}=\frac{1}{2}\int_0^1\sum_n|\dot{\hat p}_n(t)|\dd{t}.
\end{equation}
Noting that $\ell_W^\infty=(1/2)\int_0^1\sum_n|\dot{\hat p}_{0,n}(t)|\dd{t}$, we can prove as follows:
\begin{align}
	|\ell_W-\ell_W^\infty|&=\frac{1}{2}\big|\int_0^1\sum_n[|\dot{\hat p}_n(t)|-|\dot{\hat p}_{0,n}(t)|]\dd{t}\big|\notag\\
	&\le\frac{1}{2}\int_0^1\sum_n|\dot{\hat p}_n(t)-\dot{\hat p}_{0,n}(t)|\dd{t}\notag\\
	&=O(\tau^{-1}).
\end{align}
Thus, $\ell_W=\ell_W^\infty+O(\tau^{-1})$.
The argument for dynamical activity can be shown as follows:
\begin{align}
	\overline{a}_\tau&=\tau^{-1}\int_0^\tau\sum_k\tr{L_k(t)\varrho_tL_k(t)^\dagger}\dd{t}\notag\\
	&=\int_0^1\sum_k\tr{L_k(\hat{\vb*{\Lambda}}_t)[\pi_{\hat{\vb*{\Lambda}}_t}+O(\tau^{-1})]L_k(\hat{\vb*{\Lambda}}_t)^\dagger}\dd{t}\notag\\
	&=\int_0^1\sum_k\tr{L_k(\hat{\vb*{\Lambda}}_t)\pi_{\hat{\vb*{\Lambda}}_t}L_k(\hat{\vb*{\Lambda}}_t)^\dagger}\dd{t}+O(\tau^{-1})\notag\\
	&=\overline{a}_\tau^\infty+O(\tau^{-1}).
\end{align}
Using these asymptotic expressions and the trade-off relation \eqref{eq:power.eff.tradeoff} in the main text, we readily obtain the following geometric bound on efficiency in the adiabatic regime:
\begin{equation}\label{eq:app.tradeoff.adia}
	\eta\le\eta_C\qty(1-\frac{2T_c}{\overline{a}_\tau^\infty}\frac{(\ell_W^\infty)^2}{\mca{W}\tau}),
\end{equation}
where $\mca{W}\coloneqq\int_0^1\tr{H(\hat{\vb*{\Lambda}}_t)\dot\pi_{\hat{\vb*{\Lambda}}_t}}\dd{t}$ is the quasistatic work.
Once the control protocol is given, the upper bound in Eq.~\eqref{eq:app.tradeoff.adia} can be evaluated without solving the dynamics.  

In Ref.~\cite{Brandner.2020.PRL}, a geometric bound on efficiency was derived for cyclic heat engines in the adiabatic regime, given by
\begin{equation}\label{eq:prev.tradeoff.adia}
	\eta\le\eta_C\qty(1-\frac{\mca{L}^2}{\mca{W}\tau}),
\end{equation}
where $\mca{L}$ denotes the thermodynamic length.
Although both Eqs.~\eqref{eq:app.tradeoff.adia} and \eqref{eq:prev.tradeoff.adia} are geometric bounds, a qualitative difference between them is that our bound is in terms of only quantities at equilibriums (i.e., $\ell_W^\infty$ and $\overline{a}_\tau^\infty$), whereas the extant bound contains a response quantity from equilibriums (i.e., $\mca{L}$).
In general, there is no magnitude relationship between our bound and the extant bound.
To verify this, we calculate the two bounds for the heat engine model considered in the main text.
Note that $H_t=\varepsilon_t[\sin(\theta_t)\sigma_x+\cos(\theta_t)\sigma_z]/2=(\varepsilon_t/2)\qty[ \dyad{1_t} - \dyad{0_t} ]$, where $\ket{1_t}=[\cos(\theta_t/2),\sin(\theta_t/2)]^\top$ and $\ket{0_t}=[\sin(\theta_t/2),-\cos(\theta_t/2)]^\top$ are instantaneous energy eigenstates.
The quasistatic geometric length and dynamical activity can be calculated as follows:
\begin{align}
	\ell_W^\infty&=\frac{1}{2}\int_0^1\sum_n|\dot{\hat p}_{0,n}(t)|\dd{t}\notag\\
	&=\int_0^1\qty|\frac{d}{dt}\frac{1}{e^{\hat\beta_t\hat\varepsilon_t}+1}|\dd{t}\notag\\
	&=\int_0^1\frac{|\hat\beta_t\dot{\hat\varepsilon}_t+\dot{\hat\beta}_t\hat\varepsilon_t|e^{\hat\beta_t\hat\varepsilon_t}}{(e^{\hat\beta_t\hat\varepsilon_t}+1)^2}\dd{t},\\
	\overline{a}_\tau^\infty&=\int_0^1\sum_k\tr{L_k(\hat{\vb*{\Lambda}}_t)\pi_{\hat{\vb*{\Lambda}}_t}L_k(\hat{\vb*{\Lambda}}_t)^\dagger}\dd{t}\notag\\
	&=\int_0^1\qty[\gamma\hat\varepsilon_t(\hat n_t+1)\mel{\hat 1_t}{\pi_{\hat{\vb*{\Lambda}}_t}}{\hat 1_t}+\gamma\hat\varepsilon_t\hat n_t\mel{\hat 0_t}{\pi_{\hat{\vb*{\Lambda}}_t}}{\hat 0_t}]\dd{t}\notag\\
	&=2\int_0^1\frac{\gamma\hat\varepsilon_te^{\hat\beta_t\hat\varepsilon_t}}{e^{2\hat\beta_t\hat\varepsilon_t}-1}\dd{t}.
\end{align}
Following Ref.~\cite{Brandner.2020.PRL}, the thermodynamic length $\mca{L}$ can be calculated as
\begin{equation}
	\mca{L}=\int\sqrt{g_{\vb*{\Lambda}}^{\mu\nu}\dd{\Lambda_{\mu}}\dd{\Lambda_{\nu}}},
\end{equation}
where the metric tensor components $g_{\vb*{\Lambda}}^{\mu\nu}$ are given by
\begin{equation}
	g_{\vb*{\Lambda}}^{\mu\nu}=-\frac{R_{\vb*{\Lambda}}^{\mu\nu}+R_{\vb*{\Lambda}}^{\nu\mu}}{2}.
\end{equation}
Here, the adiabatic response coefficients $R_{\vb*{\Lambda}}^{\mu\nu}$ can be expressed as
\begin{equation}\label{eq:res.tmp}
	R_{\vb*{\Lambda}}^{\mu\nu}=-\frac{1}{T}\int_0^\infty\braket{\exp(\msf{K}_{\vb*{\Lambda}} u)\delta F_{\vb*{\Lambda}}^\mu}{\delta F_{\vb*{\Lambda}}^\nu}\dd{u},
\end{equation}
where $\braket{X}{Y}\coloneqq\int_0^1\tr{\pi_{\vb*{\Lambda}}^{1-s}X^\dagger\pi_{\vb*{\Lambda}}^{s}Y}\dd{s}$ is a scalar product and $\msf{K}_{\vb*{\Lambda}}$ denotes the adjoint Lindblad generator, defined as
\begin{equation}
	\msf{K}_{\vb*{\Lambda}}\varrho\coloneqq i[H(\vb*{\Lambda}),\varrho]+\sum_k\mca{D}^\ddag[L_k(\vb*{\Lambda})]\varrho.
\end{equation}
Here, $F_{\vb*{\Lambda}}^\mu=-\partial_\mu H(\vb*{\Lambda})$ for $\mu\in\{\varepsilon,\theta\}$, $F_{\vb*{\Lambda}}^\mu=-\ln\pi_{\vb*{\Lambda}}$ for $\mu=T$, $\delta F_{\vb*{\Lambda}}^\mu=F_{\vb*{\Lambda}}^\mu-\braket{\mbb{1}}{F_{\vb*{\Lambda}}^\mu}$, and $\mca{D}^\ddag[L]\varrho\coloneqq L^\dagger\varrho L-\{L^\dagger L,\varrho\}/2$.
The generator $\msf{K}_{\vb*{\Lambda}}$ has normalized eigenvectors $\{M_{i,\vb*{\Lambda}}\}$ and eigenvalues $\{\kappa_{i,\vb*{\Lambda}}\}$ as
\begin{equation}\label{eq:res.tmp1}
	\msf{K}_{\vb*{\Lambda}}M_{i,\vb*{\Lambda}}=\kappa_{i,\vb*{\Lambda}}M_{i,\vb*{\Lambda}},\quad \braket{M_{i,\vb*{\Lambda}}}{M_{j,\vb*{\Lambda}}}=\delta_{ij}.
\end{equation}
Specifically, $\{M_{i,\vb*{\Lambda}}\}$ and $\{\kappa_{i,\vb*{\Lambda}}\}$ are given by
\begin{align}
	M_{0,\vb*{\Lambda}}&=\mbb{1},~\kappa_{0,\vb*{\Lambda}}=0,\\
	M_{1,\vb*{\Lambda}}&=e^{\beta\varepsilon/2}\dyad{1} - e^{-\beta\varepsilon/2}\dyad{0},~\kappa_{1,\vb*{\Lambda}}=-\gamma\varepsilon\coth(\beta\varepsilon/2),\\
	M_{2,\vb*{\Lambda}}&=\sqrt{\beta\varepsilon\coth(\beta\varepsilon/2)}\dyad{1}{0},~\kappa_{2,\vb*{\Lambda}}=i\varepsilon-\frac{\gamma\varepsilon\coth(\beta\varepsilon/2)}{2},\\
	M_{3,\vb*{\Lambda}}&=\sqrt{\beta\varepsilon\coth(\beta\varepsilon/2)}\dyad{0}{1},~\kappa_{3,\vb*{\Lambda}}=-i\varepsilon-\frac{\gamma\varepsilon\coth(\beta\varepsilon/2)}{2}.
\end{align}
Note also that $\delta F_{\vb*{\Lambda}}^\mu$ can be expanded in terms of $\{M_{i,\vb*{\Lambda}}\}$ as $\delta F_{\vb*{\Lambda}}^\mu=\sum_i\braket{M_{i,\vb*{\Lambda}}}{\delta F_{\vb*{\Lambda}}^\mu}M_{i,\vb*{\Lambda}}$. By simple algebraic calculations, we can show that
\begin{align}
	\delta F_{\vb*{\Lambda}}^\varepsilon&=\frac{-1}{2\cosh(\beta\varepsilon/2)}M_{1,\vb*{\Lambda}},\\
	\delta F_{\vb*{\Lambda}}^T&=\frac{\beta\varepsilon}{2\cosh(\beta\varepsilon/2)}M_{1,\vb*{\Lambda}},\\
	\delta F_{\vb*{\Lambda}}^\theta&=\frac{1}{2}\sqrt{\frac{\varepsilon\tanh(\beta\varepsilon/2)}{\beta}}(M_{2,\vb*{\Lambda}}+M_{3,\vb*{\Lambda}}).\label{eq:res.tmp2}
\end{align}
Plugging Eqs.~\eqref{eq:res.tmp1}--\eqref{eq:res.tmp2} into Eq.~\eqref{eq:res.tmp}, the adiabatic response coefficients can be explicitly obtained as
\begin{widetext}
\begin{equation}
\begin{pmatrix}
	R_{\vb*{\Lambda}}^{\varepsilon\varepsilon} & R_{\vb*{\Lambda}}^{\varepsilon\theta} & R_{\vb*{\Lambda}}^{\varepsilon T}\\
	R_{\vb*{\Lambda}}^{\theta\varepsilon} & R_{\vb*{\Lambda}}^{\theta\theta} & R_{\vb*{\Lambda}}^{\theta T}\\
	R_{\vb*{\Lambda}}^{T\varepsilon} & R_{\vb*{\Lambda}}^{T\theta} & R_{\vb*{\Lambda}}^{TT}
\end{pmatrix} = \begin{pmatrix}
	\frac{-\beta\tanh(\beta\varepsilon/2)}{4\gamma\varepsilon\cosh(\beta\varepsilon/2)^2} & 0 & \frac{\beta^2\tanh(\beta\varepsilon/2)}{4\gamma\cosh(\beta\varepsilon/2)^2}\\
	0 & \frac{-\gamma}{[\gamma\coth(\beta\varepsilon/2)]^2+4} & 0\\
	\frac{\beta^2\tanh(\beta\varepsilon/2)}{4\gamma\cosh(\beta\varepsilon/2)^2} & 0 & \frac{-\beta^3\varepsilon\tanh(\beta\varepsilon/2)}{4\gamma\cosh(\beta\varepsilon/2)^2}
\end{pmatrix}.	
\end{equation}
Consequently, the thermodynamic length can be calculated as
\begin{equation}
	\mca{L}=\int_0^1\sqrt{\frac{\tanh(\hat\beta_t\hat\varepsilon_t/2)}{4\gamma\hat\beta_t\hat\varepsilon_t\cosh(\hat\beta_t\hat\varepsilon_t/2)^2}(\hat\beta_t\dot{\hat\varepsilon}_t+\hat\varepsilon_t\dot{\hat\beta}_t)^2+\frac{\gamma}{[\gamma\coth(\hat\beta_t\hat\varepsilon_t/2)]^2+4}\dot{\hat\theta}_t^2}\dd{t}.
\end{equation}
\end{widetext}
We plot the quantities $(2T_c/\overline{a}_\tau^\infty)(\ell_W^\infty)^2$ and $\mca{L}^2$ for a specific protocol in Fig.~\ref{fig:NumComp}.
As seen, $(2T_c/\overline{a}_\tau^\infty)(\ell_W^\infty)^2$ can be larger or smaller than $\mca{L}^2$, depending on the parameter region.
\begin{figure}[t]
\centering
\includegraphics[width=1\linewidth]{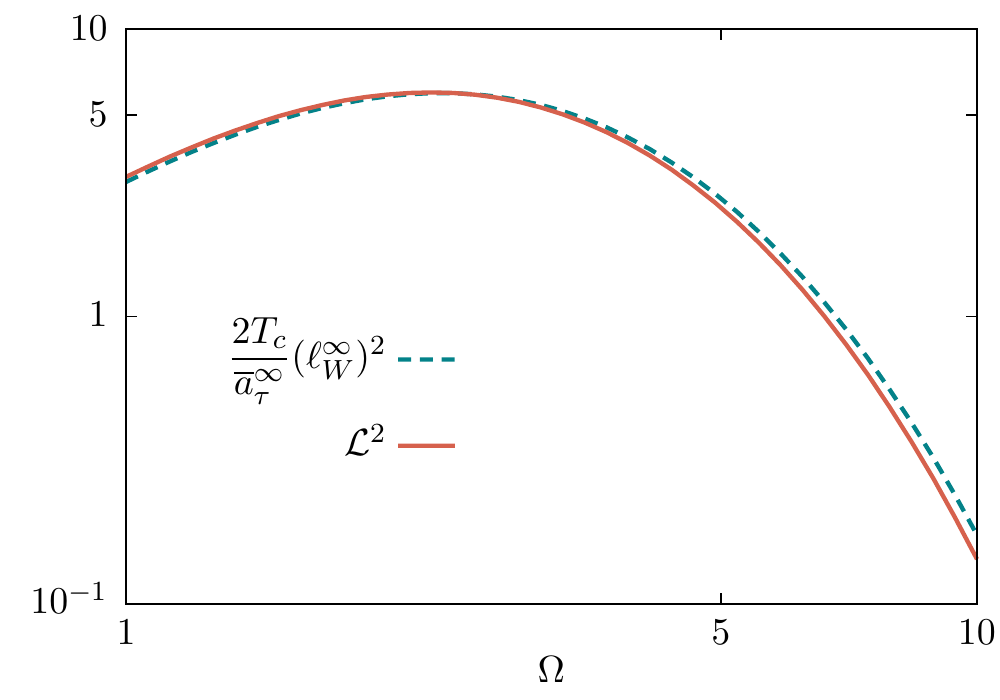}
\protect\caption{Numerical illustration of the geometric quantities $(2T_c/\overline{a}_\tau^\infty)(\ell_W^\infty)^2$ and $\mca{L}^2$ for the qubit heat engine in the adiabatic regime. The protocol is specified as $\hat\beta_t=[T_h+(T_c-T_h)\sin(\pi t)^2]/(T_hT_c)$, $\hat\varepsilon_t=\Omega [1+0.5\sin(2\pi t)+0.1\sin(6\pi t)]$, and $\hat\theta_t=0$. Parameter $\Omega$ is varied while other parameters are fixed as $T_c=1$, $T_h=1.1$, and $\gamma=0.1$.}\label{fig:NumComp}
\end{figure}
\begin{proposition}\label{prop:geo.len.exp}
For the generic case that the number of changing the magnitude order of the lines $\{p_n(t)\}^{0\le t\le\tau}_{n}$ is finite, the geometric length associated with the quantum Wasserstein distance can be expressed as
\begin{equation}
	\ell_W=\frac{1}{2}\int_0^\tau\sum_n|\dot p_n(t)|\dd{t}.
\end{equation}
\end{proposition}
\begin{proof}
Define $K=\tau/\delta t$, where $\delta t$ is an infinitesimal time step.
Let $S_K\subseteq\{0,1,\dots,K-1\}$ be the set of indices such that there exist changes in the magnitude order of the lines $\{p_n(t)\}_n$ during the time interval $[k\delta t,(k+1)\delta t]$ for any $k\in S_K$.
Then, for any $k\notin S_K$, the magnitude order of $\{p_n(t)\}_n$ is invariant for all $t\in[k\delta t,(k+1)\delta t]$.
In this case, the quantum Wasserstein distance $W(\varrho_{k\delta t},\varrho_{(k+1)\delta t})$ can be expressed as
\begin{align}
	W(\varrho_{k\delta t},\varrho_{(k+1)\delta t})&=\frac{1}{2}\sum_n|a_n((k+1)\delta t)-b_n(k\delta t)|\notag\\
	&=\frac{1}{2}\sum_n|p_n((k+1)\delta t)-p_n(k\delta t)|~\forall k\notin S_K.\label{eq:Wdis.exp}
\end{align}
Here, $\{a_n((k+1)\delta t)\}$ and $\{b_n(k\delta t)\}$ are increasing eigenvalues of $\varrho_{(k+1)\delta t}$ and $\varrho_{k\delta t}$, respectively.
Since the number of changing the magnitude order is finite, there exists a constant $C<\infty$ such that $|S_K|<C$ for any $K$.
Therefore, we can show that
\begin{align}
	0\le \delta\ell&\coloneqq\lim_{\delta t\to 0}\frac{1}{2}\sum_{k\in S_K}\sum_n|p_{n}((k+1)\delta t)-p_{n}(k\delta t)|\notag\\
	&\le\frac{1}{2}\lim_{\delta t\to 0}C\delta t \max_{k}\sum_n\frac{|p_{n}((k+1)\delta t)-p_n(k\delta t)|}{\delta t}=0. 
\end{align}
In other words, we have $\delta\ell=0$.
From the definition of the quantum Wasserstein distance, we can upper bound $\ell_W$ as follows:
\begin{align}
	\ell_W&=\lim_{\delta t\to 0}\frac{1}{2}\sum_{k=0}^{K-1}\sum_n|a_{n}((k+1)\delta t)-b_{n}(k\delta t)|\notag\\
	&\le \lim_{\delta t\to 0}\frac{1}{2}\sum_{k=0}^{K-1}\sum_n|p_{n}((k+1)\delta t)-p_{n}(k\delta t)|\notag\\
	&=\frac{1}{2}\int_0^\tau\sum_n|\dot p_n(t)|\dd{t}.\label{eq:prop1.tmp1}
\end{align}
On the other hand, using Eq.~\eqref{eq:Wdis.exp} and the equality $\delta\ell=0$, we also obtain a lower bound for $\ell_W$ as
\begin{align}
	\ell_W&=\lim_{\delta t\to 0}\frac{1}{2}\sum_{k=0}^{K-1}\sum_n|a_{n}((k+1)\delta t)-b_{n}(k\delta t)|\notag\\
	&\ge \lim_{\delta t\to 0}\frac{1}{2}\sum_{k\notin S_K}\sum_n|a_{n}((k+1)\delta t)-b_{n}(k\delta t)|\notag\\
	&=\lim_{\delta t\to 0}\frac{1}{2}\sum_{k\notin S_K}\sum_n|p_{n}((k+1)\delta t)-p_{n}(k\delta t)|\notag\\
	&=\lim_{\delta t\to 0}\Big[\frac{1}{2}\sum_{k=0}^{K-1}\sum_n|p_{n}((k+1)\delta t)-p_{n}(k\delta t)|\notag\\
	&-\frac{1}{2}\sum_{k\in S_K}\sum_n|p_{n}((k+1)\delta t)-p_{n}(k\delta t)|\Big]\notag\\
	&=\frac{1}{2}\int_0^\tau\sum_n|\dot p_n(t)|\dd{t}-\delta\ell\notag\\
	&=\frac{1}{2}\int_0^\tau\sum_n|\dot p_n(t)|\dd{t}.\label{eq:prop1.tmp2}
\end{align}
From Eqs.~\eqref{eq:prop1.tmp1} and \eqref{eq:prop1.tmp2}, we immediately obtain $\ell_W=(1/2)\int_0^\tau\sum_n|\dot p_n(t)|\dd{t}$, which completes the proof.
\end{proof}


\begin{thebibliography}{78}
\bibitem [{fnt({\natexlab{a}})}]{fnt1}%
  \BibitemOpen
  \href@noop {} {}Mathematically, $\ell_{W}= 0$ is also a solution of
  $\mca{B}=0$. However, $\ell_{W}=0$ exactly means that no state change occurs
  along the path, meaning no power output.\BibitemShut {Stop}%
\bibitem [{fnt({\natexlab{b}})}]{fnt2}%
  \BibitemOpen
  \href@noop {} {}In the adiabatic limit, the quasistatic work, geometric
  length, and dynamical activity can be explicitly calculated as
  $\mca{W}=\int_0^\tau\tr{H_t\dot\pi_t}\dd{t}$,
  $\ell_W^\infty=(1/2)\int_0^\tau\sum_n|\dot\pi_n(t)|\dd{t}$, and
  $\overline{a}_\tau^\infty=\lim_{\tau\to\infty}\tau^{-1}\int_0^\tau\sum_k\tr{L_k(t)\pi_tL_k(t)^\dagger}\dd{t}$.
  Here, $\{\pi_n(t)\}$ are eigenvalues of the instantaneous Gibbs state
  $\pi_t=e^{-\beta_tH_t}/\tr{e^{-\beta_tH_t}}$.\BibitemShut {Stop}%
\end{thebibliography}
\end{document}